\documentclass[smallextended,natbib]{svjour3}                                                                                                                                                                                                  

\smartqed  
\usepackage{graphicx}
\usepackage{wrapfig}
\newcommand\eat[1]{}
\usepackage{tikz}
\usepackage{booktabs}  
\usetikzlibrary{arrows}
\usepackage{slashbox}
\journalname{}

\usepackage{array,xspace,multirow,hhline,graphicx,xcolor,tikz,colortbl,tabularx,amsmath,amssymb,amsfonts}
\usetikzlibrary{shapes}
\usepackage{algorithm,algorithmic}
\usepackage{mathrsfs}
\usepackage{enumitem}
\setenumerate[1]{label=(\emph{\roman{*}}),ref=(\emph{\roman{*}}),leftmargin=*}

\usepackage{txfonts}
\usepackage{eucal} 

\newcommand{\ie}{i.e.,\xspace}
\newcommand{\eg}{e.g.,\xspace}

\usepackage{tikz}

\newcommand{\secref}[1]{Section~\ref{#1}}

\newcommand{\thmref}[1]{Theorem~\ref{#1}}

\newcommand{\exref}[1]{Example~\ref{#1}}

\newcommand{\set}[1]{\{#1\}}

\newcommand{\jr}{\ensuremath{\mathit{JR}}\xspace}
\newcommand{\ejr}{\ensuremath{\mathit{EJR}}\xspace}
\newcommand{\pjr}{\ensuremath{\mathit{PJR}}\xspace}

\newcommand{\strong}{strong justified representation\xspace}
\newcommand{\semi}{semi-strong justified representation\xspace}
\newcommand{\semimix}{(semi-)strong justified representation\xspace}

\newcommand{\pav}[0]{\ensuremath{\mathit{PAV}}\xspace}
\newcommand{\mav}[0]{\ensuremath{\mathit{MAV}}\xspace}

\newcommand{\av}[0]{\ensuremath{\mathit{AV}}\xspace}	
\newcommand{\sav}[0]{\ensuremath{\mathit{SAV}}\xspace}
\newcommand{\rav}[0]{\ensuremath{\mathit{RAV}}\xspace}

\newcommand{\grav}[0]{\ensuremath{\mathit{GAV}}\xspace}

\newcommand{\ccav}[0]{\ensuremath{\mathit{CCAV}}\xspace}
\newcommand{\monav}[0]{\ensuremath{\mathit{MonAV}}\xspace}

\newcommand{\calA}{{\vec{A}}}
\newcommand{\calG}{{\cal{G}}}
\newcommand{\vecw}{{\mathbf{w}}}
\newcommand{\vecu}{{\mathbf{u}}}
\newcommand{\vecx}{{\mathbf{x}}}
\newcommand{\vecy}{{\mathbf{y}}}
\newcommand{\eps}{{\varepsilon}}

\let\enumtemp=\enumerate
\def\enumerate{\enumtemp\itemsep 1pt}
\let\itemtemp=\itemize
\def\itemize{\itemtemp\itemsep 1pt}

\newcommand{\Omit}[1]{}

\begin{document}

	\title{Justified Representation in \\Approval-Based Committee Voting
	}


	\author{Haris Aziz \and Markus Brill \and Vincent Conitzer \and Edith Elkind \and Rupert Freeman \and Toby Walsh}

	\institute{%
	  H. Aziz \and 
	  T. Walsh \at
	 Data61, CSIRO and UNSW Australia,
	 	  Sydney 2052 , Australia \\
	 	  Tel.: +61-2-8306\,0490 \\
	 	  Fax: +61-2-8306\,0405 \\
	 \email{\{haris.aziz,toby.walsh\}@data61.csiro.au}
	\and 
	  M. Brill \and 
	  E. Elkind \at
  		  University of Oxford\\
  		  Oxford OX1 3QD, UK\\
		  \email{\{mbrill,elkind\}@cs.ox.ac.uk}
	\and
	  V. Conitzer \and 
 	  R. Freeman \at
		  Duke University\\Durham, NC 27708, USA\\
		  \email{\{conitzer, rupert\}@cs.duke.edu}
}
	
	\date{Received: date / Accepted: date}

\maketitle

\begin{abstract}
We consider approval-based committee voting, i.e.~the setting where each voter approves 
a subset of candidates, and these votes are then used to select a fixed-size set of winners (committee). 
We propose a natural axiom for this setting, which we call {\em justified representation (\jr)}. 
This axiom requires that if a large enough group of voters exhibits agreement
by supporting the same candidate, then at least one voter in this group
has an approved candidate in the winning committee. We show that for every list of ballots 
it is possible to select a committee that provides \jr. However, it turns out 
that several prominent approval-based voting rules may fail to output such a committee.
In particular, while Proportional Approval Voting (\pav) always outputs
a committee that provides \jr, Reweighted Approval Voting (\rav), 
a tractable approximation to \pav, does not have this property.
We then introduce a stronger version of the \jr axiom, 
which we call {\em extended justified representation (\ejr)},
and show that \pav satisfies \ejr, while other rules we consider do not;
indeed, \ejr can be used to characterize \pav within the class of weighted \pav rules. 
We also consider several other questions related to \jr and \ejr, including 
the relationship between \jr/\ejr and core stability, 
and the complexity of the associated algorithmic problems.
\end{abstract}



%

\section{Introduction}

The aggregation of preferences is a central problem in the field of social choice.
While the most-studied scenario is that of selecting a single candidate out of many, 
it is often the case that one needs to select a fixed-size {\em set of winners (committee)}:
this includes domains such as parliamentary elections, 
the hiring of faculty members, or (automated) agents deciding on a set of plans~\citep{LMM07a,DOS14a,EFSS14a,ELS11a,SFL15}.
The study of algorithmic complexity of voting rules that output committees 
is an active research direction 
\citep{PSZ08a,MPRZ08b,CKM10a,LuBo11d,cor-gal-spa:c:sp-width,BSU13a,SFS13a,sko-yu-fal:c:mwsc}.

In this paper we consider approval-based rules, where each voter lists
the subset of candidates that she approves of. There is a growing literature on voting 
rules that are based on approval ballots: the Handbook on Approval Voting \citep{LS10}
provides a very useful survey of pre-2010 research on this topic, and after this seminal book was published,  
various aspects of approval voting continued to attract a considerable amount of attention
\citep[see, \eg the papers of][]{CKM10a,Endr13a,Dud14a}.
One of the advantages of approval ballots is their simplicity: 
such ballots reduce the cognitive burden on voters (rather 
than providing a full ranking of the candidates, a voter only needs to decide which candidates to approve) and are also easier to communicate to the election authority.  
The most straightforward way to aggregate approvals is to have every approval for a candidate 
contribute one point to that candidate's score and select the candidates with the highest score. This rule is called \emph{Approval Voting (\av)}.
\av has many desirable properties in the single-winner case \citep{BMS06a,Endr13a},
including its ``simplicity, propensity to elect Condorcet winners (when they exist),
its robustness to manipulation and its monotonicity'' \citep[p.~viii]{Bram10a}. However, for the case 
of multiple winners, the merits of $\av$  are ``less clear'' \citep[p.~viii]{Bram10a}. 
For example, \av may fail proportional representation: if the goal is to select $k$ winners, $k>1$,
$51\%$ of the voters approve the same $k$ candidates, and the remaining
voters approve a disjoint set of $k$ candidates, then the voters in minority 
do not get any of their approved candidates selected.

As a consequence, over the years, several multi-winner rules based on approval ballots 
have been proposed \citep[see, \eg the survey by][]{Kilg10a}; we will now briefly describe the rules that 
will be considered in this paper (see Section~\ref{sec:prelim} for formal definitions). 
Under \emph{Proportional Approval Voting (\pav)}, each voter's contribution to the committee's 
total score depends on how many candidates from the voter's approval set have been elected.
In the canonical variant of this rule the marginal utility of the $\ell$-th approved candidate is 
$\frac{1}{\ell}$, i.e.~this rule is associated with the weight vector $(1,\frac12,\frac13,\dots)$;
other weight vectors can be used as well, resulting in the family of 
{\em weighted \pav rules}. A sequential variant of \pav
is known as {\em Reweighted Approval Voting (\rav)}; again, by varying the weight vector,
we obtain the family of {\em weighted \rav rules}. 
Another way to modulate the approvals is through computing a satisfaction
score for each voter based on the ratio of the number of their approved candidates appearing
in the committee and their total number of approved candidates; this idea leads to \emph{Satisfaction Approval Voting (\sav)}.
One could also use a distance-based approach: 
\emph{Minimax Approval Voting (\mav)} selects a set of $k$ candidates that minimizes the maximum
Hamming distance from the submitted ballots. 
Finally, one could adapt classic rules that provide fully proportional representation, such 
as the Chamberlin--Courant rule \citep{ChCo83a} or the Monroe rule \citep{Monr95a}, to work
with approval ballots, by using each voter's ballot as a scoring vector. 
All the rules informally described above
have a more egalitarian objective than \av. For example, Steven Brams, a proponent of \av in single-winner elections,
has argued that \sav is more suitable for equitable representation in multi-winner elections \citep{BrKi14a}. 

The relative merits of approval-based multi-winner rules 
and the complexity of winner determination under these rules have been examined in great detail in both economics 
and computer science in recent years \citep{BrFi07c,LMM07a,MPRZ08b,CKM10a,AGG+14a,BS14,MNS15}. 
On the other hand, there has been 
limited axiomatic analysis of these rules from the perspective of representation
(see, however, Section~\ref{sec:related}).

In this paper, we introduce the notion of {\em justified representation} (\jr) in approval-based voting. 
Briefly, a committee is said to provide justified representation for a given set of ballots if every large
enough group of voters with shared preferences is allocated at least one representative. A rule is said 
to satisfy justified representation if it always outputs a committee that provides justified 
representation. This concept is 
related
to the \emph{Droop proportionality criterion} \citep{Dro81}
and Dummett's \emph{solid coalition property} \citep{Dumm84a,TiRi00a,EFSS14a}, but is specific to approval-based elections. 

We show that every set of ballots admits a committee that provides justified representation; moreover, such a committee can be computed
efficiently, and checking whether a given committee provides \jr can be done in polynomial time as well. 
This shows that justified representation is a reasonable requirement. 
However, it turns out that many popular multi-winner approval-based rules fail \jr; 
in particular, this is the case for \av, \sav, \mav and the canonical variant of \rav.
On the positive side, \jr is satisfied by some of the weighted \pav rules, including the canonical
\pav rule, as well as by the weighted \rav rule associated with the weight vector $(1,0,\dots)$
and by the Monroe rule. Also, \mav satisfies \jr for a restricted domain of voters' preferences. 
We then consider a strengthening of the \jr axiom, which we call 
{\em extended justified representation (\ejr)}. This axiom captures the intuition that a very large group of voters 
with similar preferences may deserve not just one, but several representatives. \ejr turns out to be a more demanding property
than \jr: of all voting rules considered in this paper, only the canonical \pav rule satisfies \ejr. 
Thus, in particular, \ejr characterizes the canonical \pav rule within the class of weighted \pav rules.
However, we show that it is computationally hard to check whether a given committee provides \ejr. 

We also consider other strengthenings of \jr, which we call {\em \semi} and {\em \strong}; 
however, it turns out that for some inputs the requirements imposed by these axioms are impossible to satisfy.
Finally, we explore the relationship between \jr/\ejr and core stability in a non-transferable
utility game that can be associated with a multiwinner approval voting scenario. We show that,
even though \ejr may appear to be similar to core stability, it is, in fact, a strictly weaker
condition. Indeed, the core stability condition appears to be too demanding, as none
of the voting rules considered in our work is guaranteed to produce a core stable outcome,
even when the core is known to be non-empty.
We conclude the paper by showing how \jr can be used to formulate other attractive approval-based multi-winner rules, 
discussing related work, and identifying several directions for future work.

\section{Preliminaries}\label{sec:prelim}
We consider a social choice setting with a set $N=\{1,\ldots, n\}$ of voters and a set $C$ of candidates. 
Each voter $i\in N$ submits an approval ballot $A_i\subseteq C$, which represents the subset of candidates that she
approves of. We refer to the list $\calA = (A_1,\ldots, A_n)$ of approval ballots as the {\em ballot profile}. 
We will consider {\em approval-based multi-winner voting rules} that take as input a tuple $(N, C, \calA, k)$, 
where $k$ is a positive integer that satisfies $k\le |C|$, and return a subset 
$W \subseteq C$ of size $k$, which we call the {\em winning set}, or {\em committee} \citep{MaKi12a}. 
We omit $N$ and $C$ from the notation when they are clear from the context.
Several approval-based multi-winner rules are defined below.
Whenever the description of the rule does not uniquely specify a winning set, we assume that ties are broken
according to some deterministic procedure; however, most of our results do not depend on 
the tie-breaking rule.

\subsection{Approval-Based Multi-Winner Rules}

\noindent{\bf Approval Voting (AV)\ }
Under \av, the winners are the $k$ candidates that receive the largest number of approvals. Formally, 
the {\em approval score} of a candidate $c\in C$ is defined as $|\{i\mid c\in A_i\}|$, and \av outputs a set $W$ of size $k$
that maximizes $\sum_{c\in W}|\{i\mid c\in A_i\}|$.
\av has been adopted by several academic and professional societies such as 
the Institute of Electrical and Electronics Engineers (IEEE) and 
the International Joint Conference on Artificial Intelligence (IJCAI).

\smallskip

\noindent{\bf Satisfaction Approval Voting (SAV)\ }
A voter's {\em satisfaction score} is the fraction of her approved
candidates that are elected. \sav maximizes the sum of voters' satisfaction scores. Formally,
\sav  outputs a set $W \subseteq C$ of size $k$ that maximizes $\sum_{i\in N}\frac{|W\cap A_i|}{|A_i|}$.
This rule was proposed with the aim of ``representing more diverse interests'' than~$\av$ \citep{BrKi14a}.

\smallskip

\noindent{\bf Proportional Approval Voting (PAV)\ }
Under \pav, a voter is assumed to derive a utility of $1+\frac12+\frac13+ \dots +\frac1{j}$
from a committee that contains exactly $j$ of her approved candidates, 
and the goal is to maximize the sum of the voters' utilities.
Formally, the \pav-score of a set $W \subseteq C$ 
is defined as $\sum_{i\in N}r(|W\cap A_i|)$, where $r(p)=\sum_{j=1}^p\frac{1}{j}$,
and $\pav$ outputs a set $W \subseteq C$ of size $k$
with the highest \pav-score.
Though sometimes attributed to Forest Simmons, \pav was already proposed by the
Danish polymath Thorvald N. Thiele in the 19th century \citep{Thie95a}.\footnote{We 
are grateful to Xavier Mora and Svante Janson for pointing this out to us.}
\pav captures the idea of diminishing returns: an individual voter's preferences should count less the more she is satisfied. 

We can generalize the definition of \pav by using an arbitrary 
score vector in place of $(1, \frac12, \frac13, \dots)$. Specifically, for every vector\footnote{
	It is convenient to think of $\vecw$ as an infinite vector; note that for an election with $m$ candidates only the first 
	$m$ entries of $\vecw$ matter. To analyze the complexity of $\vecw$-\pav rules, one would have to place additional 
	requirements on $\vecw$; however, we do not consider algorithmic properties of such rules in this paper.}
$\vecw=(w_1, w_2, \dots)$, 
where $w_1, w_2, \dots$ are non-negative reals, 
we define a voting rule $\vecw$-\pav that operates as follows. 
Given a ballot profile $(A_1, \dots, A_n)$
and a target number of winners $k$, $\vecw$-\pav returns a set $W$ of size $k$ with the highest
$\vecw$-\pav score, defined by $\sum_{i\in N}r_\vecw(|W\cap A_i|)$, where $r_\vecw(p)=\sum_{j=1}^p w_j$.
Usually, it is required that $w_1=1$ and $w_1\ge w_2\ge\dots$. The latter constraint is 
appropriate in the context of representative democracy: it is motivated by the intuition that
once an agent already has one or more representatives in the committee, that agent 
should have less priority for further representation. 
It what follows, we will always impose the constraint $w_1=1$ (as we can always rescale the weight vector,
this is equivalent to requiring that $w_1>0$; while the case $w_1=0$ may be of interest in some applications, we omit it in order to keep the length of the paper manageable\footnote{
Generalizations of \pav with $w_1=0$ have been considered by \citet{FiPe04a} and \citet{SFL15a}. We note that such rules do not satisfy justified representation (as defined in \secref{sec:jr}).
}) 
and explicitly
indicate which of our results require that $w_1\ge w_2\ge\dots$; in particular,
for our characterization of \pav in \thmref{thm:wpav-not-ejr} this is not the case.

\smallskip

\noindent{\bf Reweighted Approval Voting (RAV)\ }
\rav converts \pav into a multi-round rule, by selecting a candidate
in each round and then reweighing the approvals for the subsequent rounds.
Specifically, \rav starts by setting $W=\emptyset$.
Then in round $j$, $j=1, \dots, k$, it computes the {\em approval weight}
of each candidate $c$ as $\sum_{i: c\in A_i}\frac{1}{1+|W\cap A_i|}$, 
selects a candidate with the highest approval weight, and adds him to $W$.
After $k$ rounds, it outputs the set $W$.
\rav has also been referred to 
as ``\emph{sequential proportional AV}'' \citep{BrKi14a},
and was used briefly in Sweden during the early 1900s.

\citet{Thie95a} proposed \rav as a tractable approximation to \pav
(see Section~\ref{sec:complex} for a discussion of the computational complexity 
of these rules and the relationship between them). We note that 
there are several other examples of voting rules that were conceived
as approximate versions of other rules, yet became viewed as legitimate voting rules
in and of themselves; two representative examples are the Simplified Dodgson rule of \citet{Tid}, 
which was designed as an approximate version of the Dodgson rule \citep[see the discussion by][]{CKKP14},
and the Greedy Monroe rule of \citet{SFS13a}, which approximates the Monroe rule~\citep{Monr95a}.

Just as for \pav, we can extend the definition of \rav to score vectors other than $(1, \frac12, \frac13, \dots)$:
every vector $\vecw=(w_1, w_2, \dots)$ 
defines a sequential 
voting rule $\vecw$-\rav, which proceeds as \rav, except that it computes 
the approval weight of a candidate $c$ in round $j$ as $\sum_{i: c\in A_i}w_{|W\cap A_i|+1}$, 
where $W$ is the winning set after the first $j-1$ rounds.
Again, we impose the constraint $w_1=1$ (note that if $w_1=0$, then $\vecw$-\rav
can pick an arbitrary candidate at the first step, which is obviously undesirable).

A particularly interesting rule in this class is $(1, 0, \dots)$-\rav: 
this rule, which we will refer to as {\em Greedy Approval Voting (\grav)},
can be seen as a variant of the SweetSpotGreedy (SSG) algorithm of \citet{LuBo11d},
and admits a very simple description: we pick candidates one by one, 
trying to `cover' as many currently `uncovered' voters as possible.
In more detail, a winning committee under this rule can be computed by the following algorithm. 
We start by setting $C'=C$, $\calA'=\calA$, and $W=\emptyset$. As long as $|W|<k$ and $\calA'$ is non-empty,
we pick a candidate $c\in C'$ that has the highest approval score with respect to $\calA'$,
and set $W:=W\cup\{c\}$, $C':=C'\setminus\{c\}$. Also, we remove from $\calA'$ all ballots $A_i$ such that
$c\in A_i$. If at some point we have $|W|<k$ and $\calA'$ is empty, we add an arbitrary set of  $k-|W|$
candidates from $C'$ to $W$ and return $W$; if this does not happen,
we terminate after having picked $k$ candidates.

We will also consider a variant of \grav, where, at each step, 
after selecting a candidate $c$,
instead of removing all voters in $\calA_c=\{i\mid c\in A_i\}$ from $\calA'$, 
we remove a subset of $\calA_c$ of size $\min\left\{\lceil\frac{n}{k}\rceil, |\calA_c|\right\}$.
This rule can be seen as an adaptation of the classic STV rule to approval ballots,
and we will refer to it as $\grav^T$ (where $T$ stands for `threshold').%
\footnote{For readability, we use the Hare quota $\lceil\frac{n}{k}\rceil$; however, all our proofs 
go through if we use the Droop quota $\lceil\frac{n}{k+1}\rceil+1$ instead.
For a discussion of differences between these two quotas, see the article of \citet{Tide95a}.}

\smallskip

\noindent{\bf Minimax Approval Voting (MAV)\ }
\mav returns a committee $W$ that minimizes the maximum \emph{Hamming distance} between $W$ and the voters' ballots;
this rule was proposed by \citet{BKS07a}.
Formally, let $d(Q, T)=|Q\setminus T|+|T\setminus Q|$ 
and define the \mav-score of a set $W\subseteq C$ as
$\max \left(d(W,A_1),\ldots, d(W,A_n)\right)$.
\mav outputs a size-$k$ set with the lowest \mav-score.

\smallskip

\noindent{\bf Chamberlin--Courant and Monroe Approval Voting (CCAV and MonAV)\ }
The Chamberlin--Courant rule \citep{ChCo83a} is usually defined
for the setting where each voter provides a full ranking of the candidates.	
Each voter $i\in N$ is associated with a scoring vector $\vecu^i=(u_1^i,\dots, u_m^i)$
whose entries are non-negative reals;
we think of $u^i_j$ as voter $i$'s satisfaction from being represented by candidate 
$c_j$. A voter's satisfaction from a committee $W$ is defined as $\max_{c_j\in W}u^i_j$,
and the rule returns a committee of size $k$ that maximizes the sum of voters' satisfactions.
For the case of approval ballots, it is natural to define the scoring vectors
by setting $u^i_j=1$ if $c_j\in A_i$ and $u^i_j=0$ otherwise; that is, a voter
is satisfied by a committee if this committee contains one of her approved candidates.
Thus, the resulting rule is equivalent to $(1,0,\dots)$-\pav (and therefore we
will not discuss it separately).

The Monroe rule \citep{Monr95a} is a modification of the Chamberlin--Courant rule
where each committee member represents roughly the same number of voters.
Just as under the Chamberlin--Courant rule, we have a scoring vector $\vecu^i=(u_1^i,\dots, u_m^i)$
for each voter $i\in N$. Given a committee $W$ of size $k$, we say that a mapping
$\pi:N\to W$ is {\em valid} if it satisfies 
$|\pi^{-1}(c)|\in\left\{\lfloor\frac{n}{k}\rfloor, \lceil\frac{n}{k}\rceil\right\}$ for each $c\in W$.
The Monroe score of a valid mapping $\pi$ is given by $\sum_{i\in N} u^i_{\pi(i)}$,
and the Monroe score of a committee $W$ is the maximum Monroe score of a 
valid mapping from $N$ to $W$. 
The Monroe rule returns a size-$k$ committee with the maximum Monroe score.
For approval ballots, we define the scoring vectors in the same manner
as for the Chamberlin--Courant Approval Voting rule; we call the resulting rule the 
Monroe Approval Voting rule (\monav). 

\smallskip

We note that for $k=1$, \av,  \pav, \rav, \grav, $\grav^T$ and \monav produce the same output
if there is a unique candidate with the highest approval score. However, such a
candidate need not be a winner under \sav or \mav.

\subsection{Computational Complexity}\label{sec:complex}
The rules listed above differ from an algorithmic perspective.
For some of these rules, namely, \av, \sav, \rav, \grav and $\grav^T$, a winning committee
can be computed in polynomial time; this is also true for $\vecw$-\rav
as long the entries of the weight vector are rational numbers that can be efficiently
computed given the number of candidates. In contrast, \pav, \mav, and \monav
are computationally hard \citep{AGG+14a,SFL15,LMM07a,PSZ08a}; 
for $\vecw$-\pav, the hardness result holds for most 
weight vectors, including $(1,0, \dots)$, (\ie it holds for \ccav).
However, both \pav and \mav admit efficient approximation algorithms 
(i.e., algorithms that output committees which are approximately optimal with respect 
to the optimization criteria of these rules) and have been analyzed 
from the perspective of parameterized complexity.
Specifically, $\vecw$-\pav admits an efficient $\left(1-\frac{1}{e}\right)$-approximation algorithm 
as long as the weight vector $\vecw$ is efficiently computable and non-increasing;
in fact, such an algorithm is provided by $\vecw$-\rav \citep{SFL15}.
For \mav, \citet{LMM07a} propose a simple $3$-approximation algorithm;
\citet{CKM10a} improve the approximation ratio to $2$ and \citet{BS14}
develop a polynomial-time approximation scheme.
\citet{MNS15} show that \mav is fixed-parameter tractable for a number of natural parameters;
\citet{EL15} obtain fixed parameter tractability results for \pav
when voters' preferences are, in some sense, single-dimensional.
There is also a number of tractability results for \ccav, and, to a lesser extent, for \monav;
we refer the reader to the work of \citet{SFL15} and references therein.

\section{Justified Representation}\label{sec:jr}

We will now define one of the main concepts of this paper.
	
\begin{definition}[Justified representation (JR)]
Given a ballot profile $\calA = (A_1, \dots, A_n)$ over a candidate set $C$ and a target committee size $k$,
we say that a set of candidates $W$ of size $|W|=k$ {\em provides justified representation 
for $(\calA, k)$} if there does not exist a set of voters $N^*\subseteq N$ with $|N^*|\ge \frac{n}{k}$ such that
$\bigcap_{i\in N^*}A_i\neq \emptyset$ and $A_i\cap W=\emptyset$ for all $i\in N^*$.
We say that an approval-based voting rule {\em satisfies justified representation (\jr)} if for every profile
$\calA = (A_1, \dots, A_n)$ and every target committee size $k$ it outputs a winning set that provides
justified representation for $(\calA, k)$.
\end{definition}

The logic behind this definition is that if $k$ candidates are to be selected, then, intuitively,
each group of $\frac{n}{k}$ voters ``deserves'' a representative. Therefore, a set of $\frac{n}{k}$ voters 
that have at least one candidate in common should not be completely unrepresented.
We refer the reader to \secref{sec:scr} for a discussion of alternative definitions.


\subsection{Existence and Computational Properties}
We start our analysis of justified representation by observing that, for every ballot profile $\calA$ and every value of $k$, 
there is a committee that provides justified representation for $(\calA, k)$, and, moreover, such a committee
can be computed efficiently given the voters' ballots. In fact, both \grav and $\grav^T$
output a committee that provides \jr.

\begin{theorem}\label{th:grav}
\grav and $\grav^T$ satisfy \jr.
\end{theorem}
\begin{proof}
We present a proof that applies to both \grav and $\grav^T$.
Suppose for the sake of contradiction that for some ballot profile $\calA=(A_1, \dots, A_n)$ and some $k>0$,
\grav (respectively, $\grav^T$) outputs a committee that does not provide justified representation for $(\calA, k)$.
Then there exists a set $N^*\subseteq N$ with $|N^*|\ge \frac{n}{k}$ such that $\bigcap_{i\in N^*} A_i\neq\emptyset$
and, when \grav (respectively, $\grav^T$) terminates, every ballot $A_i$ such that $i\in N^*$ 
is still in $\calA'$. Consider some candidate
$c\in \bigcap_{i\in N^*} A_i$. At every point in the execution of our algorithm, $c$'s approval score
is at least $|N^*|\ge \frac{n}{k}$. As $c$ was not elected, at every stage the algorithm selected
a candidate whose approval score was at least as high as that of $c$. Thus, at the end of each stage  
the algorithm removed from $\calA'$ at least $\lceil\frac{n}{k}\rceil$
ballots containing the candidate added to $W$ at that stage, 
so altogether the algorithm has removed at least $k\cdot \frac{n}{k}$ ballots from $\calA'$.
This contradicts the assumption that $\calA'$ contains at least $\frac{n}{k}$ 
ballots when the algorithm terminates.
\qed
\end{proof}

Theorem~\ref{th:grav} shows that it is easy to find a committee that provides justified representation
for a given ballot profile. It is also not too hard to check whether a given committee $W$ provides \jr. Indeed, 
while it may seem that we need to consider every subset of voters of size $\frac{n}{k}$, in fact it is sufficient
to consider the candidates one by one, and, for each candidate~$c$, 
compute $s(c)=|\{i\in N\mid c\in A_i, A_i\cap W=\emptyset\}|$;
the set $W$ fails to provide justified representation for $(\calA, k)$
if and only if there exists a candidate $c$ with $s(c) \ge \frac{n}{k}$.
We obtain the following theorem.

\begin{theorem}\label{th:poly-jr}
There exists a polynomial-time algorithm that, given a ballot profile $\calA$
over a candidate set $C$, and a committee $W$, $|W|=k$, decides
whether $W$ provides justified representation for $(\calA, k)$.
\end{theorem}


\subsection{Justified Representation and Unanimity}
A desirable property of single-winner approval-based voting rules 
is {\em unanimity}: a voting rule is unanimous if, given a ballot profile $(A_1, \dots, A_n)$
with $\cap_{i\in N} A_i\neq\emptyset$, it outputs a candidate in $\cap_{i\in N}A_i$.
This property is somewhat similar in spirit to \jr, so
the reader may expect that for $k=1$ it is equivalent to \jr. 
However, it turns out that the \jr axiom is strictly weaker than unanimity for $k=1$: 
while unanimity implies \jr, the converse is not true, as illustrated by the following example.

\begin{example}\label{ex:jr-unan}
Let $N=\{1, \dots, n\}$, $C=\{a,b_1,\dots, b_n\}$, $A_i=\{a,b_i\}$ for $i\in N$.
Consider a voting rule that for $k=1$ outputs $b_1$ on this profile and coincides with \grav
in all other cases. Clearly, this rule is not unanimous; however, it satisfies \jr,
as it is impossible to find a group of $\frac{n}{k}=n$ unrepresented voters for $(A_1,\dots, A_n)$.
\end{example}

It is not immediately clear how to define unanimity for multi-winner voting rules;
however, any reasonable definition would be equivalent to the standard definition of unanimity when $k=1$,
and therefore would be different from justified representation.

We remark that a rule can be unanimous for $k=1$ and provide \jr for all values of $k$:
this is the case, for instance, for \grav.

\section{Justified Representation under Approval-Based Rules}

We have argued that justified representation is a reasonable condition: there always exists a committee that 
provides it, and, moreover, such a committee can be computed efficiently. It is therefore natural to ask
whether prominent voting rules satisfy \jr. In this section, we will answer this question for \av, \sav, \mav, \pav, \rav, and \monav. 
We will also identify conditions on $\vecw$ that are sufficient/necessary 
for $\vecw$-\pav and $\vecw$-\rav to satisfy \jr.

In what follows, for each rule we will try to identify the range of values of $k$ for which this rule satisfies \jr. Trivially, all 
rules that we consider satisfy \jr for $k=1$.
It turns out that \av fails \jr for $k> 2$, and for $k=2$ the answer depends on the tie-breaking rule.

\begin{theorem}\label{thm:av-not-jr}
For $k=2$, \av satisfies \jr if ties are broken in favor of sets that provide \jr.
For  $k\ge 3$, \av fails \jr.
\end{theorem}
\begin{proof}
Suppose first that $k=2$. Fix a ballot profile $\calA$.
If every candidate is approved by fewer than $\frac{n}{2}$ voters in $\calA$,
\jr is trivially satisfied. If some candidate is approved by more than $\frac{n}{2}$ voters in $\calA$,
then \av selects some such candidate, in which case no group of $\lceil\frac{n}{2}\rceil$
voters is unrepresented, so \jr is satisfied in this case as well. It remains to consider
the case where $n=2n'$, some candidates are approved by $n'$ voters, and no candidate
is approved by more than $n'$ voters. Then \av necessarily picks at least one candidate
approved by $n'$ voters; denote this candidate by $c$. In this situation \jr can only be violated
if the $n'$ voters who do not approve $c$ all approve the same candidate (say, $c'$),
and this candidate is not elected. But the approval score of $c'$ is $n'$, and,
by our assumption, the approval score of every candidate is at most $n'$, so this
is a contradiction with our tie-breaking rule. This argument also illustrates
why the assumption on the tie-breaking rule is necessary: it can be the case
that $n'$ voters approve $c$ and $c''$, and the remaining $n'$ voters approve $c'$,
in which case the approval score of $\{c, c''\}$ is the same as that of $\{c,c'\}$.

For $k\ge 3$, we let $C=\{c_0, c_1,\dots,c_k\}$, $n=k$,
and consider the profile where the first voter approves $c_0$,
whereas each of the remaining voters approves all of $c_1,\dots, c_k$.
\jr requires $c_0$ to be selected, but \av selects $\{c_1,\dots,c_k\}$. 
\qed
\end{proof}

On the other hand, \sav and \mav fail \jr even for $k=2$.

\begin{theorem}\label{thm:sav-not-jr}
\sav and \mav  do not satisfy \jr for $k\ge 2$.
\end{theorem}
\begin{proof}
We first consider \sav. Fix $k\ge 2$, 
let $X=\{x_1,\dots,x_k, x_{k+1}\}$, $Y=\{y_1,\dots,y_k\}$, $C=X\cup Y$, 
and consider the profile $(A_1, \dots, A_k)$,
where $A_1=X$, $A_2=\{y_1,y_2\}$, $A_i=\{y_i\}$ for $i=3,\dots, k$.
\jr requires each voter to be represented, but \sav will choose $Y$:
the \sav-score of $Y$ is $k-1$, whereas the \sav-score of every committee
$W$ with $W\cap X\neq\emptyset$ is at most $k-2+\frac12+\frac{1}{k+1}<k-1$. 
Therefore, the first voter will remain unrepresented.

For \mav, we use the following construction. Fix $k\ge 2$,
let $X=\{x_1,\dots,x_k\}$, $Y=\{y_1,\dots,y_k\}$, $C=X\cup Y\cup\{z\}$,
and consider the profile $(A_1, \dots, A_{2k})$,
where $A_i=\{x_i, y_i\}$ for $i=1,\dots, k$, $A_i=\{z\}$ for $i=k+1,\dots,2k$.
Every committee of size $k$ that provides \jr for this profile contains~$z$.
However, \mav fails to select $z$. Indeed, the \mav-score of $X$ is $k+1$:
we have $d(X,A_i)=k$ for $i\le k$ and $d(X,A_i)=k+1$ for $i>k$.
Now, consider some committee $W$ with $|W|=k$, $z\in W$. We have
$A_i\cap W=\emptyset$ for some $i\le k$, so $d(W, A_i)=k+2$. 
Thus, \mav prefers $X$ to any committee that includes $z$.
\qed
\end{proof}
The constructions used in the proof of Theorem~\ref{thm:sav-not-jr} show 
that \mav and \sav may behave very differently: \sav appears to favor voters 
who approve very few candidates, whereas \mav appears to favor voters 
who approve many candidates.

Interestingly, we can show that \mav satisfies \jr if we assume that each voter
approves exactly $k$ candidates and ties are broken in favor of sets that provide \jr.
	
\begin{theorem}\label{th:mav-resr-jr}
If the target committee size is $k$, $|A_i|=k$ for all $i\in N$,
and ties are broken in favor of sets that provide \jr, then \mav satisfies \jr .
\end{theorem}	
\begin{proof}
Consider a profile $\calA=(A_1, \dots, A_n)$ with $|A_i|=k$ for all $i\in N$.

Observe that if there exists a set of candidates $W$ with $|W|=k$
such that $W\cap A_i\neq\emptyset$ for all $i\in N$, then \mav will necessarily select
some such set. Indeed, for any such set $W$ we have $d(W, A_i)\le 2k-1$
for each $i\in N$, whereas if $W'\cap A_i=\emptyset$ for some set $W'$ with $|W'|=k$
and some $i\in N$, then $d(W',A_i)=2k$. Further, by definition, 
every set $W$ such that $|W|=k$ and $W\cap A_i\neq\emptyset$ for all $i\in N$ 
provides justified representation for $(\calA, k)$. 

On the other hand, if there is no $k$-element set of candidates that intersects
each $A_i$, $i\in N$, then the \mav-score of every set of size $k$ is $2k$,
and therefore \mav can pick an arbitrary size-$k$ subset. Since we assumed
that the tie-breaking rule favors sets that provide \jr, our claim follows.   
\qed
\end{proof}

While Theorem~\ref{th:mav-resr-jr} provides an example of a setting
where \mav satisfies \jr, this result is not entirely
satisfactory: first, we had to place a strong restriction on voters' preferences,
and, second, we used a tie-breaking rule that was tailored to \jr.

We will now show that \pav satisfies \jr, for all ballot profiles and irrespective of the tie-breaking rule.

\begin{theorem}\label{thm:pav-jr}
\pav satisfies \jr.
\end{theorem}
\begin{proof}
Fix a ballot profile $\calA = (A_1, \dots, A_n)$ and a $k> 0$ and let $s=\lceil \frac{n}{k}\rceil$.
Let $W$ be the output of \pav on $(\calA,k)$.
Suppose for the sake of contradiction that there exists
a set $N^*\subset N$, $|N^*|\ge s$, such that $\bigcap_{i\in N^*}A_i\neq\emptyset$, 
but $W\cap \bigcup_{i\in N^*}A_i=\emptyset$.
Let $c$ be some candidate approved by all voters in $N^*$.
                        
For each candidate $w\in W$, define its {\em marginal contribution} as
the difference between the \pav-score of $W$ and that of $W\setminus\{w\}$.
Let $m(W)$ denote the sum of marginal contributions of all candidates in $W$.
Observe that if $c$ were to be added to the winning set, this would
increase the \pav-score by at least~$s$. Therefore, it suffices to
argue that the marginal contribution of some candidate in $W$ is less
than $s$: this would mean that swapping this candidate with $c$
increases the \pav-score, a contradiction. To this end, we will prove
that $m(W)\le s(k-1)$; as $|W|=k$, our claim would then follow by the pigeonhole principle.
                        
Consider the set $N\setminus N^*$; we have $n\le sk$, so $|N\setminus N^*| \le n-s \le s(k-1)$.
Pick a voter $i\in N\setminus N^*$, and let $j=|A_i\cap W|$. If $j>0$,  
this voter contributes exactly $\frac{1}{j}$ to the marginal contribution of each candidate in $A_i\cap W$, 
and hence her contribution to $m(W)$ is exactly $1$.
If $j=0$, this voter does not contribute to $m(W)$ at all.
Therefore, we have $m(W)\leq |N\setminus N^*|\le s(k-1)$, which is what we wanted to prove.
\qed
\end{proof}

The reader may observe that the proof of Theorem~\ref{thm:pav-jr} applies to 
all voting rules of the form $\vecw$-\pav where the weight vector satisfies $w_1=1$, $w_j\le \frac{1}{j}$ for all $j\ge 1$.
In Section~\ref{sec:ejr} we will see that this condition on $\vecw$ is also necessary for $\vecw$-\pav to satisfy \jr.

Next, we consider \rav. As this voting rule can be viewed as a tractable approximation of \pav 
(recall that \pav is NP-hard to compute), one could expect that \rav satisfies \jr as well. However, this turns out not 
to be the case, at least if $k$ is sufficiently large. 

\begin{theorem}\label{thm:rav-not-jr}
\rav satisfies \jr for $k=2$, but fails it for $k \ge 10$.
\end{theorem}
\begin{proof}
For $k=2$, we can use essentially the same argument as for \av;
however, we do not need to assume anything about the tie-breaking rule.	
This is because if there are three candidates, $c$, $c'$, and $c''$, such that
$c$ and $c''$ are approved by the same $\frac{n}{2}$ voters, whereas $c'$ is approved
by the remaining $\frac{n}{2}$ voters, and \rav selects $c$
in the first round, then in the second round \rav favors $c'$ over $c''$.

Now, suppose that  $k=10$.
Consider a profile over a candidate set $C=\{ c_1, \dots , c_{11}\}$ with $1199$ 
voters who submit the following ballots:
\begin{align*}
\!\!\!81 \times &\{c_1,c_2\},		&81 \times &\{c_1,c_3\},	&80 \times &\{c_2\},		&80 \times &\{c_3\},\\
\!\!\!81 \times &\{c_4,c_5\},		&81 \times &\{c_4,c_6\}, 	&80 \times &\{c_5\},		&80 \times &\{c_6\},\\
\!\!\!49 \times &\{c_7,c_8\},		&49 \times &\{c_7,c_9\},	&49 \times &\{c_7,c_{10}\}, 	& &\\	
\!\!\!96 \times &\{c_8\}, 		&96 \times &\{c_9\},		&96 \times &\{c_{10}\},		&\!\!\!\!120 \times &\{c_{11}\}.
\end{align*}

Candidates $c_1$ and $c_4$ are each approved by $162$  voters, the most of any candidate, and these blocks of $162$ voters do not overlap, 
so \rav selects $c_1$ and $c_4$ first. This reduces the \rav scores of $c_2, c_3, c_5$ and $c_6$ from 
$80+81=161$ to $80 + 40.5 = 120.5$, so $c_7$, whose \rav score is $147$, is selected next. 
Now, the \rav scores of $c_8, c_9$ and $c_{10}$ become $96 + 24.5 = 120.5$. 
The selection of any of $c_2, c_3, c_5,c_6,c_8,c_9$ or $c_{10}$ does not affect the \rav 
score of the others, so all seven
of these candidates will be selected before $c_{11}$,
who has $120$ approvals. 
Thus, after the selection of $10$ candidates, there are $120 >  \frac{1199}{10} = \frac{n}{k}$ 
unrepresented voters who jointly approve $c_{11}$.

To extend this construction to $k > 10$, we create $k-10$ additional candidates and $120(k-10)$ 
additional voters such that for each new candidate, there are $120$ new voters who approve 
that candidate only. Note that we still have $120 > \frac{n}{k}$. \rav will proceed to select 
$c_1, \dots, c_{10}$, followed by $k-10$ additional candidates, 
and $c_{11}$ or one of the new candidates will remain unselected.
\qed
\end{proof}
While \rav itself does not satisfy \jr, 	
one could hope that this can be fixed by tweaking the weights, 
i.e.~that $\vecw$-\rav satisfies \jr for a suitable weight vector $\vecw$. 
However, it turns out that $(1, 0, \dots)$ is essentially the only weight
vector for which this is the case: Theorem~\ref{thm:rav-not-jr} extends
to $\vecw$-\rav for {\em every} weight vector $\vecw$ with 
$w_1=1$, $w_2>0$. 

\begin{theorem}\label{thm:wrav-not-jr}
For every vector $\vecw=(w_1,w_2,\dots )$ with $w_1=1$, $w_2>0$, 
there exists a value of $k_0>0$ such that $\vecw$-\rav does not satisfy \jr for $k>k_0$.
\end{theorem}
\begin{proof}
Pick a positive integer $s\ge 8$ such that $w_2\ge \frac{1}{s}$.
Let $C=C_1\cup C_2\cup\{x,y\}$, where
$$
C_1 = \{c_{i, j}\mid i=1,\dots,2s+3,j=1,\dots,2s+1\}, \qquad
C_2 =\{c_i\mid i=1, \dots, 2s+3\}.
$$
For each $i=1, \dots, 2s+3$ and each $j=1,\dots, 2s+1$ we construct
$2s^3-s$ voters who approve $c_{i,j}$ only and
$s^2$ voters who approve $c_{i,j}$ and $c_i$ only.
Finally, we construct $2s^3-1$ voters who approve $x$ only and $s^2-7s-5$ voters who approve $y$ only
(note that the number of voters who approve $y$ is positive by our choice of $s$).

Set $k_0=(2s+2)(2s+3)=|C_1\cup C_2|$.
Note that the number of voters $n$ is given by
\begin{align*}
(2s+3)&(2s+1)(2s^3+s^2-s)+(2s^3-1)+(s^2-7s-5) \\
&= (2s+2)(2s+3)(2s^3-1)=(2s^3-1)k_0,
\end{align*}
and hence $\frac{n}{k_0}=2s^3-1$.

Under $\vecw$-\rav initially the score of each candidate in $C_2$ is $s^2(2s+1)=2s^3+s^2$,
the score of each candidate in $C_1$ is $2s^3+s^2-s$,
the score of $x$ is $2s^3-1$, and the score of $y$ is $s^2-7s-5$, so in the first $2s+3$
rounds the candidates from $C_2$ get elected.
After that, the score of every candidate in $C_1$
becomes $2s^3-s+w_2s^2\ge 2s^3-s+s=2s^3$, while the scores of $x$ and $y$ remains unchanged.
Therefore, in the next $(2s+3)(2s+1)$ rounds the candidates from $C_1$ get elected.
At this point, $k$ candidates are elected, and $x$ is not elected, even though
the $2s^3-1=\frac{n}{k_0}$ voters who approve him do not approve of any of the candidates in the winning set.

To extend this argument to larger values of $k$, we proceed as in the proof of Theorem~\ref{thm:rav-not-jr}:
for $k>k_0$, we add $k-k_0$ new candidates, and for each new candidate we construct $2s^3-1$ new voters
who approve that candidate only. Let the resulting number of voters be $n'$; we have $\frac{n'}{k}=2s^3-1$,
so $\vecw$-\rav will first select the candidates in $C_2$, followed by the candidates in $C_1$, and then it will choose
$k-k_0$ winners among the new candidates and $x$. As a result, either $x$ or one of the new candidates will remain unselected.
\qed
\end{proof}

\begin{remark}
Theorem~\ref{thm:wrav-not-jr} partially subsumes Theorem~\ref{thm:rav-not-jr}:
it implies that \rav fails \jr, but the proof only shows that
this is the case for $k\ge 18\cdot 19=342$, while
Theorem~\ref{thm:rav-not-jr} states that \rav fails \jr for $k\ge 10$ already.
We chose to include the proof of Theorem~\ref{thm:rav-not-jr} because we feel that
it is useful to know what happens for relatively small values of $k$.
Note, however, that Theorem~\ref{thm:rav-not-jr} leaves open the question
of whether \rav satisfies \jr for $k=3, \dots, 9$.
Very recently, \citet{SFFB16a} have answered this question by showing 
that \rav satisfies \jr for $k\le 5$ and fails it for $k\ge 6$.
\end{remark}

\noindent 
If we allow the entries of the weight vector to depend on the number of voters $n$,
we can obtain another class of rules that provide justified representation:
the argument used to show that \grav satisfies \jr extends to $\vecw$-\rav 
where the weight vector $\vecw$ satisfies
$w_1=1$, $w_j\le\frac{1}{n}$ for $j>1$. In particular, the rule $(1, \frac{1}{n}, \frac{1}{n^2}, \dots, )$-\rav
is somewhat more appealing than \grav:
for instance, if $\bigcap_{i\in N}A_i=\{c\}$ and $k>1$, \grav will pick $c$, 
and then behave arbitrarily, whereas $(1, \frac{1}{n}, \frac{1}{n^2}, \dots, )$-\rav will also
pick $c$, but then it will continue to look for candidates 
approved by as many voters as possible.

\medskip
We conclude this section by showing that \monav satisfies \jr.

\begin{theorem}\label{thm:monav-jr}
\monav satisfies \jr.
\end{theorem}
\begin{proof}
Fix a ballot profile $\calA = (A_1, \dots, A_n)$ and a $k>0$.
Let $W$ be an output of \monav on $(\calA,k)$. 
If $A_i\cap W\neq\emptyset$ for all $i\in N$,
then $W$ provides justified representation for $(\calA, k)$.
Thus, assume that this is not the case, i.e.~there exists
some voter $i$ with $A_i\cap W=\emptyset$. 
Consider a valid mapping $\pi:N\to W$ whose Monroe score
equals the Monroe score of $W$, let $c=\pi(i)$, and
set $s=|\pi^{-1}(c)|$; note that 
$s\in\{\lfloor\frac{n}{k}\rfloor, \lceil\frac{n}{k}\rceil\}$.

Suppose for the sake of contradiction 
that $W$ does not provide justified representation for $(\calA, k)$. 
Then by our choice of $s$ there exists a set $N^*\subset N$, $|N^*|=s$, 
such that $\bigcap_{i\in N^*}A_i\neq\emptyset$,
but $W\cap (\bigcup_{i\in N^*}A_i)=\emptyset$.
Let $c'$ be some candidate approved by all voters in $N^*$,
and set $W'=(W\setminus\{c\})\cup\{c'\}$. To obtain a contradiction,
we will argue that $W'$ has a higher Monroe score than $W$.

To this end, we will modify $\pi$ by first swapping the voters in $N^*$
with voters in $\pi^{-1}(c)$ and then assigning the voters in $N^*$ to $c'$.
Formally, let $\sigma:\pi^{-1}(c)\setminus N^*\to N^*\setminus \pi^{-1}(c)$ 
be a bijection between 
$\pi^{-1}(c)\setminus N^*$ and $N^*\setminus \pi^{-1}(c)$. 
We construct a mapping $\hat{\pi}: N\to W'$ by setting 
$$
\hat{\pi}(i)=
\begin{cases}
\pi(\sigma(i)) &\text{ for } i\in \pi^{-1}(c)\setminus N^*,\\
c'             &\text{ for } i\in N^*,\\
\pi(i)         &\text{ for } i\not\in \pi^{-1}(c)\cup N^*.
\end{cases}
$$
Note that $\hat{\pi}$ is a valid mapping: we have $|\hat{\pi}^{-1}(c')|=s$
and $|\hat{\pi}^{-1}(c'')|=|\pi^{-1}(c'')|$ for each $c''\in W'\setminus\{c'\}$.
Now, let us consider the impact of this modification on the Monroe score.
The $s$ voters in $N^*$ contributed nothing to the Monroe score of $\pi$,
and they contribute $s$ to the Monroe score of $\hat\pi$. By our choice 
of $c$, the voters in $\pi^{-1}(c)$ contributed at most $s-1$ to the Monroe
score of $\pi$, and their contribution to the Monroe score of $\hat\pi$
is non-negative. For all other voters their contribution
to the Monroe score of $\pi$ is equal to their contribution to the Monroe score of $\pi'$.
Thus, the total Monroe score of $\hat\pi$ is higher than that of $\pi$.
Since the Monroe score of $W$ is equal to the Monroe score of $\pi$,
and, by definition, the Monroe score of $W'$ is at least the Monroe score of $\hat\pi$,
we obtain a contradiction.~\qed
\end{proof}


\section{Extended Justified Representation}\label{sec:ejr}
We have identified four (families of) voting rules that satisfy \jr for arbitrary ballot profiles:
$\vecw$-\pav with $w_1=1$, $w_j\le \frac{1}{j}$ for $j>1$ (this class includes \pav), $\vecw$-\rav
with $w_1=1$, $w_j\le\frac{1}{n}$ for $j>1$ (this class includes \grav), $\grav^T$ and \monav.
The obvious advantage of \grav and $\grav^T$
is that their output can be computed efficiently, whereas computing the outputs of \pav or \monav is NP-hard.
However, \grav puts considerable emphasis on representing {\em every} voter, at the expense
of ensuring that large sets of voters with shared preferences are allocated an adequate number 
of representatives. This approach may be problematic in a variety of applications, such as selecting
a representative assembly, or choosing movies to be shown on an airplane, or foods to be provided at a banquet
\citep[see the discussion by][]{SFL15}. In particular, it may be desirable to have several assembly members
that represent a widely held political position, both to reflect the popularity of this position,
and to highlight specific aspects of it, as articulated by different candidates. 
Consider, for instance, the following example. 

\begin{example}\label{ex:100}
Let $k=3$, $C=\set{a,b,c,d}$, and $n=100$. One voter approves $c$, 
one voter approves $d$, and $98$ voters approve $a$ and $b$. 	
\grav would include both $c$ and $d$ in the winning set, whereas in many settings it would be more reasonable 
to choose both $a$ and $b$ (and one of $c$ and $d$); indeed, this is exactly what $\grav^T$ would do.
\end{example}

This issue is not addressed by the \jr axiom, as this axiom does not care if a given voter is represented
by one or more candidates. Thus, if we want to capture the intuition that large cohesive groups
of voters should be allocated several representatives, we need a stronger condition. 
Recall that \jr says that each group of $\frac{n}{k}$ voters that all approve the same candidate ``deserves''
at least one representative. It seems reasonable to scale this idea and say that, for every
$\ell>0$, each group of $\ell\cdot \frac{n}{k}$ voters that all approve the same $\ell$ candidates ``deserves''
at least $\ell$ representatives. This approach can be formalized as follows.
		
\begin{definition}[Extended justified representation (EJR)]
Given a ballot profile $(A_1, \dots, A_n)$ over a candidate set $C$, a target committee size $k$, $k\le |C|$,
and a positive integer $\ell$, $\ell\le k$, 
we say that a set of candidates $W$, $|W|=k$, {\em provides $\ell$-justified representation
for $(\calA, k)$} if there does not exist a set of voters $N^*\subseteq N$ with $|N^*|\ge \ell\cdot \frac{n}{k}$ such that
$|\bigcap_{i\in N^*}A_i|\ge \ell$, but $|A_i\cap W|<\ell$ for each $i\in N^*$;
we say that $W$ {\em provides extended justified representation (\ejr) for $(\calA, k)$} 
if it provides $\ell$-\jr for $(\calA, k)$ for all $\ell$, $1\le\ell\le k$.
We say that an approval-based voting rule {\em satisfies $\ell$-justified representation ($\ell$-\jr)} if for every profile
$\calA = (A_1, \dots, A_n)$ and every target committee size $k$ it outputs a committee that provides
$\ell$-\jr for $(\calA, k)$. Finally, we say that a rule 
{\em satisfies extended justified representation (\ejr)} if it satisfies $\ell$-\jr for all $\ell$, $1\le\ell\le k$.
\end{definition}
Observe that \ejr implies \jr, because the latter coincides with $1$-\jr. 

The definition of \ejr interprets ``a group $N^*$ deserves at least $\ell$ representatives''
as ``at least one voter in $N^*$ gets $\ell$ representatives''. Of course, 
other interpretations are also possible: for instance, we can require that each 
voter in $N^*$ is represented by $\ell$ candidates in the winning committee
or, alternatively, that the winning committee contains at least $\ell$
candidates each of which is approved by some member of $N^*$. However, the former
requirement is too strong: it differs from \jr for $\ell=1$ and in Section~\ref{sec:scr}
we show that there are ballot profiles for which commitees with this property do not exist even for $\ell=1$.
The latter approach, which was very recently proposed by \citet{SFFB16a} (see the discussion in Section~\ref{sec:related}),
is not unreasonable; in particular, it coincides with \jr for $\ell=1$. However, 
it is strictly less demanding than the approach we take: clearly, every rule
that satisfies \ejr also satisfies this condition. As it turns out (Theorem~\ref{thm:pav-ejr})
that every ballot profile admits a committee that provides \ejr, 
the \ejr axiom offers more guidance in choosing a good winning committee
than its weaker cousin, while still leaving us with a non-empty set of candidate committees to choose from.
Finally, the \ejr axiom in its present form is very similar to a core stability condition
for a natural NTU game associated with the input profile (see Section~\ref{sec:core});
it is not clear if the axiom of \citet{SFFB16a} admits a similar interpretation.

\subsection{Extended Justified Representation under Approval-Based Rules}

It is natural to ask which of the voting rules that satisfy \jr also satisfy \ejr.
\exref{ex:100} immediately shows that for \grav the answer is negative. 
Consequently, no $\vecw$-\rav rule 
such that the entries of $\vecw$ do not depend on $n$ satisfies \ejr:
if $w_2=0$, this rule is \grav and if $w_2>0$, our claim follows from Theorem~\ref{thm:wrav-not-jr}.
Moreover, \exref{ex:100} also implies that $\vecw$-\rav rules with 
$w_j \leq \frac{1}{n}$ for $j>1$ also fail \ejr. 

The next example shows that \mav fails \ejr even if each voter approves exactly $k$ candidates
(recall that under this assumption \mav satisfies \jr).

\begin{example}
Let $k=4$, $C=C_1\cup C_2\cup C_3\cup C_4$, where $|C_1|=|C_2|=|C_3|=|C_4|=4$
and the sets $C_1,C_2,C_3,C_4$ are pairwise disjoint. Let $\calA=(A_1, \dots, A_8)$,
where $A_i=C_i$ for $i=1,2,3$, and $A_i=C_4$ for $i=4,5,6,7,8$.
\mav will select exactly one candidate from each of the sets $C_1, C_2, C_3$ and $C_4$,
but \ejr dictates that at least two candidates from $C_4$ are chosen.
\end{example}
%
%
%
%
%
%


%
Further, \monav fails \ejr as well.

\begin{example}\label{ex:monav}
Let $k=4$, $C=\{c_1, c_2, c_3, c_4, a, b\}$, $N=\{1, \dots, 8\}$, 
$A_i=\{c_i\}$ for $i=1, \dots, 4$,
$A_i=\{c_{i-4}, a, b\}$ for $i=5, \dots, 8$.
\monav outputs $\{c_1, c_2, c_3, c_4\}$ on this profile, as this is the unique set of candidates
with the maximum Monroe score. Thus, every voter is represented by a single candidate, though the voters in
$N^*=\{5, 6, 7, 8\}$ ``deserve'' two candidates.
\end{example}
Example~\ref{ex:monav} illustrates the conflict between the \ejr axiom and the requirement to represent
all voters whenever possible. We discuss this issue in more detail in Section~\ref{sec:related}.

For $\grav^T$, it is not hard to construct an example where this rule fails \ejr 
for some way of breaking intermediate ties.

\begin{example}\label{ex:grav-t}
Let $N=\{1, \dots, 8\}$, $C=\{a, b, c, d, e, f\}$, 
$A_1=A_2=\{a\}$, $A_3=A_4=\{a, b, c\}$, $A_5=A_6=\{d, b, c\}$, $A_7=\{d, e\}$, $A_8=\{d, f\}$.
Suppose that $k=4$. Note that all voters in $N^*=\{3, 4, 5, 6\}$ approve $b$ and $c$, and $|N^*|=2\cdot \frac{n}{k}$. 
Under $\grav^T$, at the first step candidates $a$, $b$, $c$
and $d$ are tied, so we can select $a$ and remove voters $3$ and $4$. Next, we have to select $d$;
we can then remove voters $5$ and $6$. In the remaining two steps, we add $e$ and $f$ to the committee.
The resulting committee violates \ejr, as each voter in $N^*=\{3, 4, 5, 6\}$
is only represented by a single candidate.
\end{example}
We note that in Example~\ref{ex:grav-t} we can remove voters $1$ and $2$ after selecting $a$,
which enables us to select $b$ or $c$ in the second step and thereby obtain a committee that provides \ejr.
In fact, we were unable to construct an example where $\grav^T$ fails \ejr for {\em all} ways
of breaking intermediate ties; we now conjecture
that it is always possible to break intermediate ties in $\grav^T$
so as to satisfy \ejr. However, it is not clear if a tie-breaking rule with this property 
can be formulated in a succinct manner. Thus, $\grav^T$ does not seem particularly useful
if we want to find a committee that provides \ejr: even if our conjecture is true, we may have to explore all ways
of breaking intermediate ties.

In contrast, we will now show that \pav satisfies \ejr irrespective of the tie-breaking rule.
		
\begin{theorem}\label{thm:pav-ejr}
\pav satisfies \ejr.
\end{theorem}
\begin{proof}
Suppose that \pav violates \ejr for some value of $k$, and consider
a ballot profile $A_1, \dots, A_n$, a value of $\ell>0$ and a set of voters $N^*$,
$|N^*| = s \ge \ell\cdot \frac{n}{k}$, that witness this. Let $W$, $|W|=k$, be the winning set.
We know that at least one of the $\ell$ candidates approved by all
voters in $N^*$ is not elected; let $c$ be some such candidate.
Each voter in $N^*$ has at most $\ell-1$ representatives in $W$, 
so the marginal contribution of $c$ (if it were to be added to $W$) would be at least
$s\cdot \frac{1}{\ell} \ge \frac{n}{k}$. On the other hand, the argument in the proof 
of Theorem~\ref{thm:pav-jr} can be modified to show 
that the sum of marginal contributions of candidates in $W$ is at most $n$.

Now, consider some candidate $w\in W$ with the smallest marginal
contribution; clearly, his marginal contribution is at most $\frac{n}{k}$. 
If it is strictly less than $\frac{n}{k}$, we are done, as we can improve the
total \pav-score by swapping $w$ and $c$, a contradiction. Therefore
suppose it is exactly $\frac{n}{k}$, and therefore the marginal contribution
of each candidate in $W$ is exactly $\frac{n}{k}$. Since \pav satisfies \jr, we
know that $A_i\cap W\neq\emptyset$ for some $i\in N^*$. Pick some
candidate $w'\in W\cap A_i$, and set $W'=(W\setminus\{w'\})\cup \{c\}$.
Observe that after $w'$ is removed, adding $c$ increases
the total \pav-score by at least $(s-1)\cdot \frac{1}{\ell}+\frac{1}{\ell-1}>\frac{n}{k}$. 
Indeed, $i$ approves at most $\ell-2$ candidates in $W\setminus\{w'\}$ and
therefore adding $c$ to $W\setminus\{w'\}$ contributes at least $\frac{1}{\ell-1}$ to her
satisfaction. Thus, the \pav-score of $W'$ is higher than that of $W$, a contradiction again.
\qed
\end{proof}

Interestingly, Theorem~\ref{thm:pav-ejr} does not extend to weight vectors other than $(1, \frac12, \frac13, \dots)$:
our next theorem shows that \pav is essentially the unique $\vecw$-\pav rule that satisfies \ejr.

\begin{theorem}\label{thm:wpav-not-ejr}
For every weight vector $\vecw$ with $w_1=1$, $\vecw \neq (1, \frac12, \frac13, \dots)$,   
the rule $\vecw$-\pav does not satisfy \ejr.
\end{theorem} 

Theorem~\ref{thm:wpav-not-ejr} follows immediately from Lemmas~\ref{lem:wpav-not-ejr1}
and~\ref{lem:wpav-not-ejr2}, which are stated below.

\begin{lemma}\label{lem:wpav-not-ejr1}
Consider a weight vector $\vecw$ with $w_1=1$. 
If $w_j>\frac{1}{j}$ for some $j>1$, then $\vecw$-\pav fails \jr.
\end{lemma}
\begin{proof}
Suppose that $w_j=\frac{1}{j}+\eps$ for some $j>1$ and $\eps>0$.
Pick $k > \lceil\frac{1}{\eps j}\rceil+1$ so that $j$ divides $k$; let $t=\frac{k}{j}$.
Let $C=C_0\cup C_1\cup\dots\cup C_t$, where $C_0=\{c\}$,
$|C_1|=\dots=|C_t|=j$, and the sets $C_0,C_1, \dots, C_t$ are pairwise disjoint.
Note that $|C|=tj+1=k+1$. Also, construct $t+1$ pairwise disjoint groups of voters
$N_0, N_1, \dots, N_t$ so that $|N_0|=k$,
$|N_1|=\dots=|N_t|=j(k-1)$, and for each $i=0, 1,\dots,t$ the voters in $N_i$
approve the candidates in $C_i$ only.
Observe that the total number of voters is given by $n=k+tj(k-1)=k^2$.

We have $|N_0|=k=\frac{n}{k}$, so every committee that provides justified representation
for this profile must elect $c$. However, we claim that
$\vecw$-\pav elects all candidates in $C\setminus\{c\}$ instead.
Indeed, if we replace an arbitrary candidate in $C\setminus\{c\}$ with $c$,
then under $\vecw$-\pav the total score of our committee changes by
$$
k - j(k-1)\cdot\left(\frac{1}{j}+\eps\right) = 1 - j(k-1)\eps < 1-j\eps\left\lceil\frac{1}{\eps j}\right\rceil \le 0,
$$
i.e.~$C\setminus\{c\}$ has a strictly higher score than any committee that includes $c$.
\qed
\end{proof}

\begin{lemma}\label{lem:wpav-not-ejr2}
Consider a weight vector $\vecw$ with $w_1=1$. 
If $w_j<\frac{1}{j}$ for some $j>1$, then $\vecw$-\pav
fails $j$-\jr.
\end{lemma}
\begin{proof}
Suppose that $w_j=\frac{1}{j}-\eps$ for some $j>1$ and $\eps>0$.
Pick $k > j+\lceil\frac{1}{\eps}\rceil$.
Let $C=C_0\cup C_1$, where $|C_0|=j$, $C_1=\{c_1,\dots, c_{k-j+1}\}$
and $C_0\cap C_1 = \emptyset$.
Note that $|C|=k+1$. Also, construct $k-j+2$ pairwise disjoint groups of voters
$N_0, N_1, \dots, N_{k-j+1}$ so that $|N_0|=j(k-j+1)$,
$|N_1|=\dots=|N_{k-j+1}|=k-j$, the voters in $N_0$ approve the candidates in $C_0$ only,
and for each $i=1,\dots,k-j+1$ the voters in $N_i$ approve $c_i$ only.
Note that the number of voters is given by $n = j(k-j+1)+(k-j+1)(k-j)=k(k-j+1)$.

We have $\frac{n}{k}=k-j+1$ and $|N_0|=j\cdot \frac{n}{k}$, so every committee that
provides \ejr must select all candidates in $C_0$.
However, we claim that
$\vecw$-\pav elects all candidates from $C_1$ and $j-1$ candidates from $C_0$ instead.
Indeed, let $c$ be some candidate in $C_0$, let $c'$ be some candidate in $C_1$,
and let $W=C\setminus \{c\}$, $W'=C\setminus\{c'\}$.
The difference between the total score of $W$ and that of $W'$ is
$$
j(k-j+1)\left(\frac{1}{j}-\eps\right)-(k-j) < 1 - j\cdot \frac{1}{\eps} \cdot \eps < 1-j<0,
$$
i.e.~$\vecw$-\pav assigns a higher score to $W$. As this argument does not depend on the choice
of $c$ in $C_0$ and $c'$ in $C_1$, the proof is complete.
\qed
\end{proof}

\subsection{\jr, \ejr and Core Stability}\label{sec:core}
One can view (extended) justified representation as a stability condition,
by associating committees that provide \jr/\ejr with outcomes of a certain NTU game
that are resistant to certain types of deviations.

Specifically, given a pair $(\calA, k)$, where $\calA=(A_1,\dots, A_n)$,
we define an NTU game $\calG(\calA, k)$ with the set of players $N$ as follows. We assume that each coalition
of size $x$, $\ell\frac{n}{k}\le x < (\ell+1)\frac{n}{k}$, where $\ell\in\{1,\dots, k\}$, 
can ``purchase'' $\ell$ alternatives. Moreover, each player evaluates a 
committee of size $\ell$, $\ell\in\{1,\dots, k\}$, using the \pav utility function,
i.e.~$i$ derives a utility of $1+\frac12+\dots+\frac{1}{j}$ from a committee that contains
exactly $j$ of her approved alternatives (the argument goes through for $\vecw$-\pav
utilities, as long as $w_1\ge\dots\ge w_k>0$). 
Thus, for each coalition $S$ with $\ell\frac{n}{k}\le |S| < (\ell+1)\frac{n}{k}$ 
a payoff vector $\vecx\in{\mathbb R}^n$
is considered to be feasible for $S$ if and only if there exists a committee
$W\subseteq A$ with $|W|\le\ell$ such that $x_i=u_i(W)$ for each $i\in S$,
where $u_i(W)=1+\dots+\frac{1}{|A_i\cap W|}$. We denote the set of all payoff vectors
that are feasible for a coalition $S\subseteq N$ by $V(S)$.

We say that a coalition $S\subseteq N$ has a {\em profitable deviation} from a payoff vector 
$\vecx\in V(N)$ if there exists a payoff vector $\vecy\in V(S)$
such that $y_i>x_i$ for all $i\in S$. A payoff vector $\vecx$ is {\em stable}
if it is feasible for $N$ and no coalition $S\subseteq N$ has a profitable deviation
from it; the set of all stable payoff vectors is the {\em core} of $\calG(\calA, k)$.

The following theorem describes the relationship between \jr, \ejr,
and outcomes of $\calG(\calA, k)$.

\begin{theorem}\label{thm:core}
A committee $W$, $|W|=k$, provides justified representation for $(\calA, k)$
if and only if no coalition of size $\lceil \frac{n}{k}\rceil$ or less
has a profitable deviation from the payoff vector $\vecx$ associated with $W$.
Moreover, $W$ provides extended justified representation for $(\calA, k)$
if and only if for every $\ell\ge 0$ no coalition $N^*$ with 
$\ell\cdot\frac{n}{k}\le |N^*|<(\ell+1)\cdot\frac{n}{k}$, $|\cap_{i\in N^*}A_i|\ge \ell$
has a profitable deviation from $\vecx$.
\end{theorem}

\begin{proof}
Suppose that $W$ fails to provide justified representation for $(\calA, k)$,
i.e.~there exists a set of voters $N^*$, $|N^*|=\lceil\frac{n}{k}\rceil$,
who all approve some candidate $c\not\in W$, but none of them approves any of the candidates
in $W$. Then we have $x_i=0$ for each $i\in N^*$, and players in $N^*$
can successfully deviate: the payoff vector $\vecy$
that is associated with the committee $\{c\}$ is feasible for $N^*$
and satisfies $y_i=1$ for each $i\in N^*$.

Conversely, suppose that $W$ provides justified representation for $(\calA, k)$,
and consider a coalition $N^*$. If $|N^*|<\lceil\frac{n}{k}\rceil$,
then for every $\vecy\in V(N^*)$ we have $y_i=0$ for all $i\in N^*$,
so $N^*$ cannot profitably deviate.
On the other hand, if $|N^*|=\lceil\frac{n}{k}\rceil$, then every payoff 
vector $\vecy\in V(N^*)$ is associated with a committee of size $1$.
Hence, for every $\vecy\in V(N^*)$ we have $y_i\le 1$ for all $i\in N^*$,
and if $y_i=1$ for all $i\in N^*$, then $\cap_{i\in N^*} A_i\neq\emptyset$,
and therefore, since $W$ provides \jr, we have $x_i\ge1$ for some $i\in N^*$.

For \ejr the argument is similar. If $W$ fails to provide extended justified representation for $(\calA, k)$,
there exists an $\ell>0$ and a set of voters $N^*$, $|N^*|\ge\ell\cdot\frac{n}{k}$,
such that $|\bigcap_{i\in N^*}A_i|\ge \ell$, but $|A_i\cap W|<\ell$ for each $i\in N^*$.
Then we have $x_i<1+\dots+\frac{1}{\ell}$ for each $i\in N^*$, and players in $N^*$
can successfully deviate: if $S$ is a committee that consists of some $\ell$
candidates in $\bigcap_{i\in N^*}A_i$, then the payoff vector $\vecy$
that is associated with $S$ is feasible for $N^*$
and satisfies $y_i=1+\dots+\frac{1}{\ell}$ for each $i\in N^*$.

Conversely, suppose that $W$ provides extended justified representation for $(\calA, k)$,
and consider some $\ell\ge 0$ and some coalition $N^*$
with $\ell\cdot\frac{n}{k}\le |N^*|<(\ell+1)\frac{n}{k}$.
We have argued above that if $\ell=0$, then $N^*$ cannot profitably deviate. Thus, assume $\ell>0$.
Every payoff vector $\vecy\in V(N^*)$ is associated with a committee of size $\ell$.
Hence, for every $\vecy\in V(N^*)$ we have $y_i\le 1+\dots+\frac{1}{\ell}$ for all $i\in S$,
and if $y_i=1+\dots+\frac{1}{\ell}$ for all $i\in N^*$, then $|\cap_{i\in N^*} A_i|\ge\ell$.
Since $W$ provides \ejr, we have $x_i\ge 1+\dots+\frac{1}{\ell}$ for some $i\in N^*$.
\qed
\end{proof}

The second part of Theorem~\ref{thm:core} considers deviations by 
cohesive coalitions. The reader may wonder if it can be strengthened
to {\em arbitrary} coalitional deviations, 
i.e.~whether a committee provides \ejr \ {\em if and only if}
the associated payoff vector is in the core of $\calG(\calA, k)$.
The following example shows that this is not the case.

\begin{example}\label{ex:ejr-not-core}
Let $k=10$, $C=\{ x_1, x_2,  \ldots, x_{10}, y, z \}$, $N=\{1,2, \ldots, 20\}$, and 
\begin{align*}
	A_1=A_2=A_3 &= \{ x_1, y \},\\
	A_4=A_5=A_6 &= \{ x_1,z \},\\
	A_7=\ldots = A_{20} &= \{ x_2, \ldots, x_{10} \}.
\end{align*}
Then \pav outputs the committee $W = \{ x_1, x_2, \ldots, x_{10} \}$ for $(\calA,k)$; in particular, $W$ provides 
\ejr for $(\calA,k)$. However, the associated payoff vector $\vecx$ is not in the core, as the players in 
$\{1,2,3,4,5,6\}$, a coalition of size $3 \frac{n}{k}$, can successfully deviate: the payoff vector associated with 
$\{ x_1,y,z \}$ is feasible for $\{ 1,2,3,4,5,6 \}$ and provides a higher payoff than $\vecx$ to each of the first 
six players. We remark that the core of $\calG(\calA, k)$ is not empty: in particular, 
it contains the payoff vector associated with $\{x_1,\dots, x_8,y,z\}$.
\end{example}

It remains an open question whether the core of $\calG(\calA,k)$ is non-empty for every pair $(\calA,k)$. 
Further, while it would be desirable to have a voting rule that outputs a committee 
whose associated payoff vector is in the core whenever the core is not empty, we are not aware of any such rule:
every voting rule that fails \ejr also fails this more demanding criterion, and
Example~\ref{ex:ejr-not-core} illustrates that \pav fails this criterion as well.

%


\subsection{Computational Issues}
In Section~\ref{sec:jr} we have argued that it is easy to find a committee 
that provides \jr for a given ballot profile, 
and to check whether a specific committee provides \jr. In contrast, for \ejr these questions appear to be computationally 
difficult. Specifically, we were unable to design an efficient algorithm for computing a committee 
that provides \ejr; while \pav is guaranteed to find such a committee, computing its output is NP-hard. 
We remark, however, that when $\ell$ is bounded by a constant, we can efficiently compute a committee that provides
$\ell$-\jr, i.e.~the challenge is in handling large values of $\ell$. 

\begin{theorem}
A committee satisfying $\ell$-\jr can be computed in time polynomial in $n$ and $|C|^\ell$.
\end{theorem}
\begin{proof}
Consider the following greedy algorithm, which we will refer to as $\ell$-\grav. We start by setting
$C'=C$, $\calA'=\calA$, and $W=\emptyset$. As long as $|W| \le k - \ell$,
we check if there exists a set of candidates $\{c_1, \ldots, c_\ell \} \subset C'$ that is
unanimously approved by at least $\ell \frac{n}{k}$ voters in $\calA'$ (this can be done in time $n\cdot|C|^{\ell+1}$).
If such a set exists, we set $W:=W\cup\{c_1, \ldots, c_\ell \}$ and we remove from $\calA'$ all ballots $A_i$
such that $|A_i \cap W| \ge \ell$ (note that this includes all ballots $A_i$ with $\{c_1, \ldots c_\ell \} \subseteq A_i$).
If at some point we have $|W| \le k - \ell$ and there is no $\{ c_1, \ldots, c_\ell \}$
that satisfy our criterion or $|W| > k-\ell$, we add an arbitrary $k-|W|$
candidates from $C'$ to $W$ and return $W$; if this does not happen,
we terminate after having picked $k$ candidates.

Suppose for the sake of contradiction that for some profile $\calA=(A_1, \dots, A_n)$ and some $k>0$,
$\ell$-\grav outputs a committee that does not provide $\ell$-\jr for $(\calA, k)$.
Then there exists a set $N^*\subseteq N$ with $|N^*|\ge \ell \frac{n}{k}$ such that $|\bigcap_{i\in N^*} A_i| \ge \ell$
and, when $\ell$-\grav terminates, every ballot $A_i$ such that $i\in N^*$ is still in $\calA'$.
Consider some subset of candidates $\{c_1, \ldots, c_\ell \} \subseteq \bigcap_{i\in N^*} A_i$.
At every point in the execution of $\ell$-\grav this subset is unanimously approved by at least
$|N^*|\ge \ell \frac{n}{k}$ ballots in $\calA'$. As at least one of $\{ c_1, \ldots, c_\ell \}$
was not elected, at every stage the algorithm selected
a set of $\ell$ candidates that was approved by at least $\ell \frac{n}{k}$ ballots
(until more than $k-\ell$ candidates were selected). Since at the end of each stage
the algorithm removed from $\calA'$ all ballots containing the candidates that had been added to $W$ at that stage,
it follows that altogether the algorithm has removed at least
$\lfloor \frac{k}{\ell} \rfloor \cdot \ell \frac{n}{k} >
(\frac{k}{\ell}-1) \cdot \ell \frac{n}{k} = n-\ell\frac{n}{k}$ ballots from $\calA'$.
This is a contradiction, since we assumed that, when the algorithm terminates,
the $\ell \frac{n}{k}$ ballots $(A_i)_{i\in N^*}$
are still in $\calA'$.
\end{proof}

For the problem of checking
whether a given committee provides \ejr for a given input, we can establish a formal hardness result.

\begin{theorem}\label{th:ejr-hard}
Given a ballot profile $\calA$, a target committee size $k$, and a committee $W$, $|W|=k$,
it is {\em coNP}-complete to check whether $W$ provides \ejr for $(\calA, k)$.
\end{theorem}
\begin{proof}
It is easy to see that this problem is in coNP: 
to show that $W$ does not provide \ejr for $(\calA, k)$, 
it suffices to guess an integer $\ell$ and a set of voters $N^*$ of size at least $\ell\cdot \frac{n}{k}$ 
such that $|\bigcap_{i\in N^*}A_i|\ge \ell$, but $|A_i\cap W|<\ell$ for all $i\in N^*$.

To prove coNP-completeness, we reduce the classic {\sc Balanced Biclique} problem \citep[\mbox{[GT24]}]{gj}
to the complement of our problem. An instance of {\sc Balanced Biclique} is given by a bipartite
graph $(L, R, E)$ with parts $L$ and $R$ and edge set $E$, and an integer $\ell$; it is a ``yes''-instance
if we can pick subsets of vertices $L'\subseteq L$ and $R'\subseteq R$ so that $|L'|=|R'|=\ell$
and $(u, v)\in E$ for each $u\in L', v\in R'$; otherwise, it is a ``no''-instance.

Given an instance $\langle (L, R, E), \ell\rangle$ of {\sc Balanced Biclique} with $R=\{v_1, \dots, v_s\}$, 
we create an instance of our problem as follows. Assume without loss of generality that $s\ge 3$, $\ell\ge 3$.
We construct $4$ pairwise disjoint sets of candidates $C_0$, $C_1$, $C'_1$, $C_2$, 
so that $C_0=L$, $|C_1|=|C'_1|=\ell-1$, $|C_2|=s\ell+\ell-3s$, and set $C=C_0\cup C_1\cup C'_1\cup C_2$.
We then construct $3$ sets of voters $N_0$, $N_1$, $N_2$, so that $N_0=\{1, \dots, s\}$, 
$|N_1|=\ell(s-1)$, $|N_2|=s\ell+\ell-3s$ (note that $|N_2|>0$ as we assume that $\ell\ge 3$).
For each $i\in N_0$ we set $A_i=\{u_j\mid (u_j, v_i)\in E\}\cup C_1$, 
and for each $i\in N_1$ we set $A_i=C_0\cup C'_1$. The candidates in $C_2$
are matched to voters in $N_2$: each voter in $N_2$ approves exactly one candidate in $C_2$, 
and each candidate in $C_2$ is approved by exactly one voter in $N_2$.
Denote the resulting list of ballots by $\calA$. 
Finally, we set $k=2\ell-2$, and let $W=C_1\cup C'_1$.
Note that the number of voters $n$ is given by $s+\ell(s-1)+s\ell+\ell-3s=2s(\ell-1)$, so $\frac{n}{k}=s$.

Suppose first that we started with a ``yes''-instance of {\sc Balanced Biclique}, and let $(L', R')$
be the respective $\ell$-by-$\ell$ biclique. Let $C^*=L'$, $N^*=N_0\cup N_1$.
Then $|N^*|=\ell s$, all voters in $N^*$ approve all candidates in $C^*$, $|C^*|=\ell$,
but each voter in $N^*$ is only represented by $\ell-1$ candidates in $W$. Hence, $W$
fails to provide $\ell$-justified representation for $(\calA, k)$.

Conversely, suppose that $W$ fails to provide \ejr for $(\calA, k)$. That is, there exists
a value $j>0$, a set $N^*$ of $js$ voters and a set $C^*$ of $j$ candidates so that all voters
in $N^*$ approve of all candidates in $C^*$, but for each voter in $N^*$
at most $j$ of her approved candidates are in $W$. Note that, since $s>1$, we have $N^*\cap N_2=\emptyset$.
Further, each voter in $N\setminus N_2$ is represented by $\ell-1$ candidates in $W$,
so $j\ge \ell$. As $N^*=js\ge \ell s\ge s$, it follows that $|N^*\cap N_0|\ge \ell$, $|N^*\cap N_1|>0$.
Since $N^*$ contains voters from both $N_0$ and $N_1$, it follows that $C^*\subseteq C_0$.
Thus, there are at least $\ell$ voters in $N^*\cap N_0$ who approve the same $j\ge \ell$ candidates
in $C_0$; any set of $\ell$ such voters and $\ell$ such candidates corresponds to an $\ell$-by-$\ell$ biclique
in the input graph.
\qed
\end{proof}


\section{Variants of Justified Representation}\label{sec:scr}

The definition of \jr requires that if there is a group of $\lceil\frac{n}{k}\rceil$
voters who jointly approve some candidate, then the elected committee has to contain
at least one candidate approved by {\em some} member of this group. This condition
may appear to be too weak; it may seem more natural to require that \emph{every} group member 
approves some candidate in the committee, or---stronger yet---that the committee contains
at least one candidate approved by {\em all} group members. This intuition is captured by the following definitions.

\begin{definition}
Given a ballot profile $(A_1, \dots, A_n)$ and a target committee size $k$,  
we say that a committee $W$ of size $k$ provides
\begin{itemize}
	\item {\em semi-strong justified representation for $(\calA, k)$} if for each group $N^*\subseteq N$ 
        with $|N^*|\ge \frac{n}{k}$ and $\bigcap_{i\in N^*}A_i\neq \emptyset$ it holds that 
        $W\cap A_i \neq \emptyset$ for all $i \in N^*$.

	\item {\em strong justified representation for $(\calA, k)$} if for each group $N^*\subseteq N$ 
        with $|N^*|\ge \frac{n}{k}$ and $\bigcap_{i\in N^*}A_i\neq \emptyset$
        it holds that $W \cap \left( \cap_{i \in N^*} A_i \right) \neq \emptyset$.
\end{itemize}
\end{definition}

By definition, a committee providing strong justified representation 
also provides semi-strong justified representation, and a committee providing  
semi-strong justified representation also provides (standard) 
justified representation.

However, it turns out that satisfying these stronger requirements is not always feasible: 
there are ballot profiles for which no committee provides semi-strong justified representation.

\begin{example}\label{ex:strong}
	Let $k=3$ and consider the following profile with $n=9$ and $C=\set{a,b,c,d}$.
	\begin{align*}
		&A_1 = A_2 = \set{a} &&A_3=\set{a,b} &&A_4 = \{b\} &&A_5=\set{b,c}\\
		&A_6 = \{c\}         &&A_7=\set{c,d} &&A_8 = A_9 = \set{d}&
	\end{align*}
For each candidate $x \in C$, there are $\frac{n}{k}=3$ voters such that $\cap_i A_i = \set{x}$, 
and at least one of those voters has $A_i = \set{x}$. Thus, a committee that satisfies
semi-strong justified representation would have to contain all four candidates, which is impossible. 
\end{example}
%
%
While Example~\ref{ex:strong} shows that no approval-based voting rule can always
find a committee that provides strong or semi-strong justified representation, 
it may be interesting to identify voting rules that output such committees whenever they exist.

Finally, we remark that strong justified representation does not imply \ejr. 

\begin{example}\label{ex:sjr-not-ejr}
Let $C=\{a, b, c, d, e\}$, $n=4$, $k=4$, and consider the following ballot profile.
\begin{align*}
A_1: \{a,b\}&&
A_2: \{a,b\}&&
A_3: \{c\}&&
A_4: \{d, e\}
\end{align*}
\ejr requires that we choose both $a$ and $b$, but $\{a, c, d, e\}$ provides 
strong justified representation.
\end{example}

\section{Related Work}\label{sec:related}
It is instructive to compare \jr and \ejr to alternative approaches towards fair representation,
such as 
{\em representativeness}
\citep{Dud14a} and {\em proportional justified representation} \citep{SFFB16a}.


\citet{Dud14a} proposes the notion of {\em representativeness}, which applies to probabilistic voting rules.
The property Duddy proposes is incomparable with \jr:
in situations he considers ($k=2$, $n$ voters approve $x$, $n+1$ voters
approve $y$ and $z$), \jr requires that one of $y$ and $z$ should be selected, whereas Duddy
requires $x$ to be selected with positive probability. Both are
reasonable requirements, but they address different concerns. Duddy's
axiom say nothing about situations where voters are split equally
(say, $n$ voters approve $\{x,y\}$, $n$ voters approve $\{z, t\}$), whereas \jr
requires that each voter is represented. Another obvious difference
is that he allows for randomized rules.

Very recently (after the conference version of our paper was published), 
\citet{SFFB16a} came up with the notion of \emph{proportional 
justified representation (\pjr)}, which can be seen as an alternative to \ejr.
A committee is said to provide \pjr for a ballot profile $(A_1, \dots, A_n)$ over a candidate set $C$
and a target committee size $k$ if, for
every positive integer $\ell$, $\ell\le k$, there does not exist a set of voters $N^*\subseteq N$
with $|N^*|\ge \ell\cdot \frac{n}{k}$ such that $|\bigcap_{i\in
N^*}A_i|\ge \ell$, but $|(\bigcup_{i \in N^*} A_i) \cap W|<\ell$.
In contrast to $\ejr$, the \pjr condition does not require one of the voters in
$N^*$ to have $\ell$ representatives. Rather, a committee provides $\pjr$ as long as it contains $\ell$ candidates that are
approved by (possibly different) voters in $N^*$, for every group $N^*$ satisfying the size and cohesiveness constraints.
An attractive feature of \pjr is that it is compatible with the idea of {\em perfect 
representation}: a committee $W$ provides perfect representation for a group of $n$ voters
and a target committee size $k$ if $n=ks$ for some positive integer $s$ 
and the voters can be split into $k$ pairwise disjoint groups $N_1, \dots, N_k$ of size $s$ each
in such a way that there is a one-to-one mapping $\mu:W\to \{N_1, \dots, N_k\}$
such that for each candidate $a\in W$ all voters in $\mu(a)$ approve $a$.
\citet{SFFB16a} prove that every committee that provides perfect representation also provides \pjr;
in contrast, \ejr may rule out all committees that provide perfect representation, 
as illustrated by Example~\ref{ex:monav}.
It is easily seen that \pjr is a weaker requirement than \ejr, and a stronger one than \jr. 
Interestingly, \citet{SFFB16a} show that many results that we have established for \ejr also hold
for \pjr: in particular, $\vecw$-\rav violates \pjr for every weight
vector $\vecw$, and $\vecw$-\pav satisfies $\pjr$ if and only if $\vecw = (1, \frac12, \frac13, \dots)$. 

\section{Conclusions}\label{sec:concl}

		\begin{table}[h!]
				\centering
					\scalebox{1}{
			\centering
			\begin{tabular}{llll}
			\toprule
				&\jr&\ejr&Complexity\\
					Rule&&\\
		\midrule
		\av&--&--&in P\\
		\sav &--&--&in P\\

		\mav&--&--&NP-hard\\
				\rav&--&--&in P\\
		\grav&\checkmark&--&in P\\
        $\grav^T$&\checkmark&--${}^*$&in P\\
		 \monav&\checkmark&--&NP-hard\\
		 		\pav &\checkmark&\checkmark&NP-hard\\
				\bottomrule
			\end{tabular}
			}
		\caption{Satisfaction of \jr and \ejr and complexity of approval-based voting rules;
			the superscript `$*$' indicates that the rule fails the respective axiom for some way of breaking intermediate ties.
			}
			\label{table:summary:av-rep}
			\end{table}

\noindent We have formulated a desirable property of approval-based committee selection rules,
which we called justified representation (\jr). 
While \jr is fairly easy to satisfy,
it turns out that many well-known approval-based rules fail it. A prominent exception is the \pav
rule, which also satisfies a stronger version of this property, namely extended justified
representation (\ejr). Indeed, \ejr characterizes \pav within the class of $\vecw$-\pav rules,
and we are not aware of any other natural voting rule that satisfies \ejr irrespective of the tie-breaking rule
(of course, we can construct voting rules that differ from \pav, yet satisfy \ejr,
by modifying the output of \pav on profiles on which \ejr places no constraints on the output).
Perhaps the most pressing open question suggested by our work is whether
there is an efficient algorithm for finding a committee that provides \ejr for a given profile.
In particular, we would like to understand whether we can break ties in the execution of $\grav^T$
to always produce such a committee, and whether some tie-breaking rule
with this property is polynomial-time computable.    
Also, it would be interesting to see if \ejr, in combination with other natural axioms, can be used to axiomatize \pav.
Concerning \semimix, an interesting algorithmic problem is whether there are efficient algorithm 
for checking the existence of committees satisfying these requirements. 

Justified representation can also be used to formulate new approval-based rules. 
We mention two rules that seem particularly attractive:
\begin{quote}
The \emph{utilitarian $\mathit{(E)JR}$ rule} returns a committee that, among all committees that satisfy $\mathit{(E)JR}$, has the highest \av score.  \\ 
The \emph{egalitarian $\mathit{(E)JR}$ rule} returns a committee that, among all committees that satisfy $\mathit{(E)JR}$, 
 maximizes the number of representatives of the voter who has the least number of representatives in the winning committee. 
\end{quote}	
The computational complexity of winner determination for these rules 
is an interesting problem. 

Since \pav is NP-hard to compute, our study also provides additional motivation
for the use of approximation and parameterized algorithms to compute \pav outcomes.
Finally, analyzing the compatibility of \jr with 
other important properties, such as, \eg strategyproofness for dichotomous preferences, 
is another avenue of future research. 
	

\begin{acknowledgements}
The authors thank the anonymous reviewers of 
\emph{Multidisciplinary Workshop on Advances in Preference Handling (MPREF 2014)}, 
and the  \emph{Twenty-Ninth AAAI Conference (AAAI 2015)} for their helpful feedback on earlier versions of the paper. 
We further thank Martin Lackner and Piotr Skowron for valuable discussions.	
		
Brill, Conitzer, and Freeman were supported by NSF and ARO under
grants CCF-1101659, IIS-0953756, CCF-1337215, W911NF-12-1-0550, and
W911NF-11-1-0332, by a Feodor Lynen research fellowship of the
Alexander von Humboldt Foundation, and by COST Action IC1205 on Computational Social Choice. 
Brill and Elkind were partially supported by ERC-StG 639945.
Walsh also receives support from the Asian Office of Aerospace Research and Development (AOARD
124056) and the German Federal Ministry for Education and Research through the Alexander
von Humboldt Foundation.	
\end{acknowledgements}	

\normalsize

%

\bibliographystyle{plainnat}

\end{document}